\documentclass [10pt]{article}

\usepackage[dvips,pagebackref,colorlinks=true,linkcolor=blue, citecolor=blue]{hyperref}
\usepackage{amsfonts}
\usepackage{amsmath}
\usepackage{amssymb}
\usepackage{algorithmic}
\usepackage{algorithm}
\usepackage{wrapfig}

\usepackage[left=1in,top=1in,right=1in,bottom=1in,nohead]{geometry}

\def\note#1{}

\newcommand\R{{\mathbb R}}

\newcommand\p{{\mathbb Pr}}

\def\B{\{0,1\}}

\def\epsilon{\varepsilon}
\def\phi{\varphi}

\usepackage{amssymb}
\usepackage{amsfonts} \usepackage{amsmath} \usepackage{amsthm} \usepackage{bm}
\usepackage{latexsym} \usepackage{epsfig}
\usepackage{hyperref}
\newtheorem*{theorem*}{Theorem}

\newtheorem{theorem}{Theorem}
\newtheorem{lemma}[theorem]{Lemma}
\newtheorem{corollary}[theorem]{Corollary}
\newtheorem{proposition}[theorem]{Proposition}

\newcommand{\copyableTheorem}[2]{
 \newtheorem*{newthm#1}{Theorem \ref{thm#1}}
 \begin{newthm#1}
   {#2}
 \end{newthm#1}
 \expandafter\newcommand\expandafter{\csname thm#1\endcsname}{
   \begin{theorem}
     \label{thm#1}
     {#2}
   \end{theorem}
 }
}
\renewcommand{\P}[1]{{\mathbb{P}}\left[{#1}\right]}

\newcommand{\E}[1]{{\mathbb{E}}\left[{#1}\right]}

\newcommand{\CondE}[2]{{\mathbb{E}}\left[{#1}\middle\vert{#2}\right]}

\newcommand{\Var}[1]{{\mbox{Var}}\left[{#1}\right]}
\newcommand{\Cov}[1]{{{\mbox{Cov}}\left[{#1}\right]}}

\def\bp{{\bf p}}
\def\bfy{{\bf y}}

\def\br{{\bf r}}

\def\I{{\bf I}}
\def\M{{\bf m}}
\def\D{{\bf d}}
\def\A{{\bf A}}
\def\T{{\bf T}}
\def\Y{{\bf y}}
\def\C{{\bf C}}
\def\P{{\bf P}}
\def\Q{{\bf Q}}

\def\W{{\bf W}}
\def\q{{\bf q}}
\def\B{{\bf b}}
\def\mZ{{\bf Z}}
\newcommand{\tr}{\top}
\newcommand{\ones}{{\bf 1}}
\newcommand{\diag}{{\mbox{Diag}}}

\begin{document}
\title{Social Learning in a Changing World} \author{Rafael
  M. Frongillo\footnote{University of California Berkeley, Computer
    Science Department.  Supported by the National Defense Science \&
    Engineering Graduate Fellowship (NDSEG) Program.}~, Grant
  Schoenebeck\footnote{Princeton University, Computer Science
    Department.  This work was supported by Simons Foundation
    Postdoctoral Fellowship and National Science Foundation Graduate
    Fellowship.} and Omer Tamuz\footnote{Weizmann Institute. Supported
    by ISF grant 1300/08. Omer Tamuz is a recipient of the Google
    Europe Fellowship in Social Computing, and this research is
    supported in part by this Google Fellowship.}} \date{\today}
\maketitle

\begin{abstract}
  We study a model of learning on social networks in dynamic
  environments, describing a group of agents who are each trying to
  estimate an underlying state that varies over time, given access to
  weak signals and the estimates of their social network neighbors.

  We study three models of agent behavior. In the \emph{fixed response}
  model, agents use a fixed linear combination to incorporate
  information from their peers into their own estimate. This can be
  thought of as an extension of the DeGroot model to a dynamic
  setting.  In the \emph{best response} model, players calculate minimum
  variance linear estimators of the underlying state.

  We show that regardless of the initial configuration, fixed response
  dynamics converge to a steady state, and that the same holds for
  best response on the complete graph. We show that best response
  dynamics can, in the long term, lead to estimators with higher
  variance than is achievable using well chosen fixed responses.

  The \emph{penultimate prediction} model is an elaboration of the best
  response model. While this model only slightly complicates the
  computations required of the agents, we show that in some cases it
  greatly increases the efficiency of learning, and on complete
  graphs is in fact optimal, in a strong sense.

\end{abstract}

\section{Introduction}
\label{'section introduction'}
The past three decades have witnessed an immense effort by the
computer science and economics communities to model and understand
people's behavior on social networks~\cite{Jackson:06}. A particular
goal has been the study of how people share information and learn from
each other; learning from peers has been repeatedly shown to be a
driving force of many economic and social processes
(cf.~\cite{conley2001social, bandiera2006social, kohler1997learning,
  besley1994diffusion}).

\subsection{Classical approaches and results}

Early work by DeGroot~\cite{degroot1974reaching} considered a set of
agents, connected by a social network, that each have a prior belief:
a distribution over the possible values of an {\em underlying state of
  the world} - say the market value of some company. The agents
iteratively observe their neighbors' beliefs and update their own by
averaging the distributions of their neighbors. Since DeGroot, a
plethora of models for social learning have been proposed and studied.

DeGroot's simple averaging of neighbors' beliefs may seem naive and
arbitrary; economists often opt for {\em rational} models instead.
In rational models the agents update their belief not by a fixed rule,
but in an attempt to maximize a utility function. It is often assumed
that agents are {\em Bayesian}: they assume some prior distribution on
the underlying state and on other agents' behavior, have access to
some observations, and maximize the expected value of their utility,
using Bayes' Law. Bayesian social learning has a wide literature, with
noted work by Aumann~\cite{aumann1976agreeing} and the related {\em
  common knowledge} work (cf.~\cite{geanakoplos1992common}), as well
as McKelvey and Page~\cite{mckelvey1986common}, Parikh and
Krasucki~\cite{parikh1990communication}, Bala and
Goyal~\cite{BalaGoyal:96}, Gale and Kariv~\cite{GaleKariv:03}, and
many others.

Aumann~\cite{aumann1976agreeing} and
Geanakoplos~\cite{geanakoplos1982we} show that a group of Bayesian
agents, who each have an initial estimate of an {\em underlying
  state}, and repeatedly announce their estimate (in particular,
expected value) of this state, will eventually converge to the same
estimate. McKelvey and Page~\cite{mckelvey1986common} extend this
result to processes in which ``survey results'', rather than all the
estimates, are repeatedly shared. The social network in these models
is the {\em complete network}; indeed, it seems that non-trivial
dynamics and results are achieved already for this (seemingly simple)
topology. Aaronson~\cite{Aaronson:05} studies the complexity of the
computations required of the agents, again with highly non-trivial
results.

\subsection{Rationality and bounded rationality}
The term {\em rational} in economic theory refers to any behavior that
maximizes (or even attempts to maximize) some utility function. This
is in contrast to, for example, behavior that is heuristic or
fixed. Bayesian rationality optimizes in a probabilistic framework
that includes a prior and observations, and is, as mentioned above, a
commonly used paradigm.

The disadvantage of fully rational, Bayesian models is that the
calculations required of the agents can very quickly become
intractable, making their applicability to real-world settings
questionable; this tension between rationality and tractability is an
old recurring theme in behavioral economics models
(cf.~\cite{Simon:82}).

A solution often advocated is {\em bounded rationality}.  Agents still
act optimally in bounded rationality models, but only optimize with
respect to a restricted set of choices. This usually simplifies the
optimization problem that needs to be solved. For example, agents may
be required to disregard some of their available information or be
restricted in the manner that they calculate their strategy.  In
addition to serving the goal of more realistically modeling agents, a
usual added benefit of bounded rationality is that the analysis of the
model becomes easier.  We too follow this course.

A standard assumption in this literature is that ``actions speak
louder than words'' (cf. Smith and
S{\o}rensen~\cite{smith2000pathological}); agents do not participate
in a communication protocol intended to optimize the exchange of
information, but rather make inferences about each others' private
information by observing actions. For example, by observing the price
at which a person bids for a stock one may learn her estimate for the
future price, but yet not learn all of the information which she used
to arrive at this estimate.

\subsection{Informal statement of the model}
We consider a model where the underlying state $S$ is not a constant
number - as it is in all of the above mentioned models - but changes
with time, as prices and other economic quantities tend to do. In
particular we assume that the state $S=S(t)$ performs a random walk;
$S(0)$ is picked from some distribution, and at each iteration an
i.i.d.\ random variable is added to it.

The process commences with each agent having some estimator of
$S(0)$. We make only very weak assumptions about the joint
distribution of these estimators. Then, at each discrete time period
$t$, each agent receives an independent (and identical over time)
measurement of $S(t)$, and uses it to update its estimator. Also
available to it are the previous {\em estimates} of its neighbors on a
social network. Thus social network neighbors share their beliefs (or
rather, observe each others' actions), and information propagates
through the network.

While conceivably agents could optimally use all the information
available to them to estimate the underlying state, it appears that
such calculations are extremely complex. Instead, we explore {\em
  bounded rationality} dynamics, assuming that agents are restricted
to calculating linear combinations of their observations.  We note
that if the random walk and the measurements are taken to be Gaussian,
then the minimum variance unbiased linear estimator (MVULE) is also
the maximum likelihood estimator. A Gaussian random walk is a good
first-order approximation for many Economic processes (cf. classical
work by Bachelier~\cite{bachelier1900theorie}).

For the first part of the paper, we also require that these linear
combinations only involve the agents' neighbors' estimates from the
previous time period (and not earlier), as well as their new
measurement.  In the last model we slightly relax this
requirement.

We consider three models.  In the {\em fixed response} model each
agent, at each time period, estimates the underlying state by a {\em
  fixed} linear combination of its new measurement and the estimates
of its neighbors in the previous period. This is a straightforward
extension of the DeGroot model to our setting.

In the {\em best response model}, at each iteration, each agent
calculates the MVULE of the underlying state, based on its peers'
estimate from the previous round, together with its new
measurement. We assume here that at each iteration the agents know the
covariance matrix of their estimators. While this may seem like a
strong assumption, we note that, under some elaboration of our model,
this covariance matrix may be estimated by observing the process for
some number of rounds before updating one's estimator. Furthermore, it
seems that assumptions in this spirit - and often much stronger
assumptions - are necessary in order for agents to perform any kind of
optimization. For example, it is not rare in the literature of social
Bayesian learning to assume that the agents know the structure of the
entire social network graph (e.g.,~\cite{GaleKariv:03,
  parikh1990communication, AcemMuntDahlLobelOzd:08}).

Finally, we introduce the {\em penultimate prediction} model, which is
a simple extension of the best response model, additionally allowing
the agents to remember exactly one value from one round to the
next. While only slightly increasing the computational requirements on
the agents, this model exhibits a sharp increase in learning
efficiency.

\subsection{Informal statement of results}

While our long term goal is to understand this process on general
social network graphs, we focus in this paper on the {\em complete
  network}, which already exhibits mathematical richness.

We consider the system to be in a \emph{steady state} when the
covariance matrix of the agents' estimators is constant.  On general
graphs we show that fixed response dynamics converge to a steady
state.  On the complete graph we show that best response dynamics also
converge to a steady state. Both of these results hold regardless of
the initial conditions (i.e., the agents' estimators at time $t=0$).

We show that the steady state of best response dynamics is not
necessarily optimal; there exist fixed response dynamics in which the
agents converge to estimators which all have lower variance then the
estimators of the steady state of  best response dynamics.  This
shows that every agent can do better than the result of best response
by following a socially optimal rule; thus a certain {\em price of
  anarchy} is to be paid when agents choose the action that maximizes
their short term gain.

Finally, we show that in the penultimate prediction model, for the
complete graph, the agents learn estimators which are the optimal (in
the minimum variance sense) amongst all linear estimators, and thus
outperform those of fixed and best response dynamics.

We define a notion of ``socially asymptotic learning'': A model has
this property when the variance of the agents' steady-state estimators
tends towards the information-theoretical optimum with the number of
agents.  We show that the penultimate prediction model exhibits
socially asymptotic learning on the complete graph, while best
response and fixed response dynamics fail to do the same.

\section{Previous work}
Our model is an elaboration of models studied by DeMarzo, Vayanos and
Zwiebel~\cite{DeMarzo:03}, as well as Mossel and
Tamuz~\cite{mossel2010iterative, MosselTamuz10:arxiv}. There, the
state $S$ is a fixed number picked at time $t=0$, and each agent
receives a single measurement of it. The process thereafter is
deterministic, with each agent, at each iteration, recalculating its
estimate of $S$ based on its observation of its neighbors' estimates.

In~\cite{MosselTamuz10:arxiv} it is shown that if the agents calculate
the minimum variance unbiased linear estimator (MVULE) at every turn
(remembering all of their observations) then all the agents converge
to the optimal estimator of $S$, i.e. the average of the original
measurements. Furthermore, this happens in time that is at most $n
\cdot d$, where $d$ is the diameter of the graph.

When agents calculate estimates that are only based on their
observations from the previous round, then they do not necessarily
converge to the optimal estimator~\cite{mossel2010iterative}. In fact,
it is not known whether they converge at all.

A similar model is studied by Jadbabaie, Sandroni and
Tahbaz-Salehi~\cite{RePEc:pen:papers:10-005}. They explore a bounded
rationality setting in which agents receive new signals at each
iteration. An agent's private signals may be informative only when
combined with those of other agents, and yet their model achieves
efficient learning.

Our model is a special case of a model studied by Acemoglu, Nedic, and
Ozdaglar~\cite{acemoglu2008convergence}. They extend these models by
allowing the state to change from period to period. They don't require
the change in the state to be i.i.d, but only to have zero mean and be
independent in time. Their agents also receive a new, independent
measurement of the state at every period, which again need not be
identically distributed. They focus on a different regime than the one
we study; their main result is a proof of convergence in the case that
the variations in the state diminish with time, with variance tending
to zero.

In our model the change in the underlying state has constant variance,
as does the agents' measurement noise. This allows us to explore
steady states, in which the covariance matrix of the agents'
estimators does not change from iteration to iteration.

Our model is non-trivial already for a single agent, although here a
complete solution is simple, and can be calculated using tools
developed for the analysis of {\em Kalman
  filters}~\cite{kalman1960new}.

\section{Notation, formal models, and results}
\label{'section definitions'}
Let $[n] = \{1, 2, \ldots, n\}$ be a set of agents. Let $G=([n],E)$ be
a directed graph representing the agents' social network. We denote by
$\partial i = \{j | (i,j) \in E\}$ the neighbors of $i$, and assume
that always $i \in \partial i$.

We consider discrete time periods $t \in \{0,1,\ldots\}$. The {\em
  underlying state of the world} at time $t$, $S(t)$, is defined as
follows. $S(0)$ is a real random variable with arbitrary distribution,
and for $t>0$
\begin{align}
  \label{eq:s-def}
  S(t) = S(t-1) + X(t-1),
\end{align}
where $\E{X(t)}=0$, $\Var{X(t)}=\sigma^2$, and $\sigma$ is a parameter
of the model. The random variables $X(0),X(1),\ldots$ are
independent. Hence the underlying state $S(t)$ performs a random walk
with zero mean and standard deviation $\sigma$.

At time $t=0$ each agent $i$ receives $Y_i(0)$, an estimator of
$S(0)$. The only assumptions we make on their joint distribution is
that $\CondE{Y_i(0)}{S(0)}=S(0)$, i.e. the estimators are unbiased, and
that $\Var{Y_i(0)-S(0)}$ is finite for all $i$.

At each subsequent period $t > 0$, each agent $i$ receives $M_i(t)$,
an independent measurement of $S(t)$, defined by
\begin{align}
  \label{eq:m-def}
  M_i(t) = S(t) + D_i(t),
\end{align}
where $\E{D_i(t)}=0$, $\Var{D_i(t)}=\tau_i^2$, and the $\tau_i$'s are
parameters of the model.  Hence $D_i(t)$ is the measurement error of
agent $i$ at time $t$. Again, the random variables $D_i(t)$ are
independent.

At each period $t > 0$, each agent $i$ calculates $Y_i(t)$, agent
$i$'s estimate of $S(t)$, using the information available to it.
Precisely what information is available varies by the model (and is
defined below), but in all cases $Y_i(t)$ is a (deterministic) convex
linear combination of agent $i$'s measurements up to and including
time $t$, $\{M_i(t')|t' \leq t\}$, as well as the previous {\em
  estimates} of its social network neighbors, $\{Y_j(t')|t' < t, j
\in \partial i\}$.  Additionally, in the penultimate prediction model,
at each round $t$ each agent computes a value $R_i(t)$, and at round
$t+1$ uses this value to compute $R_i(t+1)$ and $Y_i(t+1)$.  Like
$Y_i(t)$, $R_i(t)$ is also a convex linear combination of the same
random variables.

In general, we shall assume that the agents are interested in
minimizing the expected squared error of their estimators,
$\E{(Y_i(t)-S(t))^2}$; assuming $Y_i(t)$ is unbiased (i.e.,
$\CondE{Y_i(t)}{S(t)}=S(t)$) this is equivalent to minimizing
$\Var{Y_i(t)-S(t)}$, which we refer to as the \emph{``variance of the
  estimator $Y_i(t)$.''} We shall assume throughout that the
estimators $Y_i(t)$ are indeed unbiased; we elaborate on this in the
definitions of the models below.

We shall (mostly) restrict ourselves to the case where the agents use
only their neighbors' estimates from the previous iteration, and not
from the ones before it. In these cases we write
\begin{align}
  \label{eq:y-def}
  Y_i(t) = A_i(t)M_i(t)+\sum_jP_{ij}(t) Y_j(t-1).
\end{align}
for some $A_i(t)$ and $P_{ij}(t)$ such that $P_{ij}=0$ whenever $j \not
\in \partial i$.

We will find it convenient to express such quantities in matrix form.
To that end we let $\M(t),\Y(t),\D(t)\in\R^n$ be column vectors with
entries $M_i(t), Y_i(t), D_i(t)$, and let
$\A(t),\P(t),\T\in\R^{n\times n}$ be the weight matrices, with $\A(t)
= \diag(A_1(t),\ldots,A_n(t))$, $\P = (P_{ij})_{ij}$ and
$\T=\Var{\D(t)}=\diag(\tau_1^2,\ldots,\tau_n^2)$.  Using this notation
Eq.~\eqref{eq:y-def} becomes
\begin{align}
  \label{eq:y-def-vec}
  \Y(t) = \A(t) \M(t) + \P(t) \Y(t-1).
\end{align}

We will also make use of the {\em covariance matrix}
\begin{align}
  \label{eq:C-def}
  \C(t) = \Var{\Y(t)-\ones S(t)},
\end{align}
where $\ones \in \R^n$ denotes the column vector of all ones. Hence
$C_{ij}(t) = \Cov{Y_i(t)-S(t),Y_j(t)-S(t)}$, which we refer to as the
{\em ``covariance of the estimators $Y_i(t)$ and $Y_j(t)$.''}

\subsection{Dynamics models}
\label{'section-models'}

\subsubsection{Best response}
\label{subsec:best-response}
The main model we study is the best response dynamics.  Here we assume
that at round $t$, each agent $i$ has access to $M_i(t)$, $\Y(t-1)$
and the covariance matrix for these values.  At each iteration $t$,
agent $i$ picks $A_i(t)$ and $\{P_{ij}(t)\}_j$ that minimize
$C_{ii}(t)=\Var{Y_i(t)-S(t)}$, under the constraints that (a) $P_{ij}(t)$ may be
non-zero only if $j \in \partial i$, and (b) $\CondE{Y_i(t)}{S(t)}=S(t)$,
i.e. $Y_i(t)$ is an unbiased estimator of $S(t)$.  In
Section~\ref{section:understanding-best-response} we show that these
minimizing coefficients are a deterministic function of $\C(t-1)$,
$\sigma$ and $\{\tau_i\}$. Hence we assume here that the agents know
these values.  By this definition $Y_i(t)$ is the {\em minimum variance unbiased
  linear estimator} (MVULE) of $S(t)$, given $M_i(t)$ and $\Y(t-1)$.

Note that it follows from our definitions that if the estimators
$\{Y_i(t-1)\}$ at time $t-1$ are unbiased then, in order for the
estimators at time $t$ to be unbiased, it must be the case that
\begin{align}
  \label{eq:convex}
  A_i(t)+\sum_jP_{ij}(t)=1.
\end{align}
Since at time zero the estimators are unbiased then it follows by
induction that Eq.~\eqref{eq:convex} hold for all $t>0$.

\subsubsection{Fixed response}

We shall also consider the case of estimators which are {\em fixed}
linear combinations of the agent's new measurement $M_i(t)$ and its
neighbors' estimators at time $t-1$. These we call fixed response
estimators. In this case we would have, using our matrix notation:
\begin{align}
  \label{eq:y-def-fixed}
  \Y(t) = \A \M(t) + \P \Y(t-1).
\end{align}
The matrices $\A$ and $\P$ are arbitrary matrices that satisfy the
following conditions: (a) $P_{ij}$ is positive and non-zero only if $j
\in \partial i$, and (b) $\Y_i(t)$ is a convex linear combination of
$M_i(t)$ and $\{Y_j(t-1)\}_j$. Equivalently,
$\A_i+\sum_j{P_{ij}}=1$, which is the same condition described in
Equation~\eqref{eq:convex}.

\subsubsection{Penultimate prediction}
Finally, we consider the penultimate prediction model where each
agent $i$ can remember one value, which we denote $R_i(t)$, from one
round $t$ to the next round $t+1$.  We assume that at round $t$,
each agent $i$ has access to $M_i(t)$, $\Y(t-1)$, $R_i(t-1)$ and the
covariance matrix for these values. We denote $\br(t) =
(R_1(t),\ldots,R_n(t))$.

We fix $R_i(0)= 0$, and let $R_i(t)$ be agent $i$'s MVULE of $S(t-1)$,
given $R_i(t-1)$ and $\bfy(t-1)$ (note that this is in general {\em not}
equal to $Y_i(t-1)$). $Y_i(t)$ now becomes the MVULE of $S(t)$ given
$R_i(t)$ and $M_i(t)$.

\subsection{Steady states and efficient learning}

We say that the system converges to a {\em steady state} $\C$ when
\begin{align*}
  \lim_{t \to \infty}\C(t) = \C.
\end{align*}

Assuming that agents are constrained to calculating linear
combinations of their measurements and neighbors' estimators, the
variance of the estimators $Y_i(t)$ of $S(t)$ at time $t$ can be bounded
from below by the variance of $Z_i(t)$, where we define $Z_i(t)$ to be the MVULE of $S(t)$ given the initial
estimators $\bfy(0)$, all measurements up to time $t-1$ $\{M_j(s) | j
\in [n], s < t\}$ and $M_i(t)$. We therefore define that a process
achieves {\em perfect learning} when $\Var{Y_i(t)-S(t)} = \Var{Z_i(t)
  - S(t)}$. Note that this definition is a natural one for the
complete graph and should be altered for general networks, where a
tighter lower bound exists.


If an agent were to know $S(t-1)$ exactly at time $t$, then, together
with $M_i(t)$, its minimum variance unbiased linear estimator for
$S(t)$ would be a linear combination of just $S(t-1)$ and $M_i(t)$,
because of the Markov property of $S(t)$. In this case it is easy to
show (see Proposition~\ref{prop:minimum-var-est}) that $C_{ii}(t) =
\Var{Y_i(t) - S(t)}$ would equal
$\sigma^2\tau_i^2/(\sigma^2+\tau_i^2)$. We say that a model achieves
{\em socially asymptotic learning} if for $n$ sufficiently large, as
the number of agents tends to infinity, the steady state $\C$ exists
and $C_{ii}$ tends to $\sigma^2\tau_i^2/(\sigma^2+\tau_i^2)$ for all
$i$.  We stress that this definition only makes sense in models where
the number of agents $n$ grows to infinity and therefore is
incomparable to perfect learning, which is defined for a particular
graph.


\section{Statement of the main results}
The following are our main results.  Let $\beta(t) = 1/(\ones^\tr \C(t)^{-1}\ones)$.

\copyableTheorem{BestResponseSteadyState}{ When $G$ is a complete
  graph, best-response dynamics converge to a unique steady-state, for
  all starting estimators $\bfy(0)$ and all choices of parameters
  $\{\tau_i\}$ and $\sigma$.  Moreover, the convergence is fast, in
  the sense that $-\log|\beta(t)-\beta^*| = O(t)$, where $\beta^* =
  \lim_{t \to \infty} \beta(t)$.}

\copyableTheorem{FixedResponseConverges}{ In fixed response
  dynamics, if $A_i > 0$ for all $i \in [n]$ then system converges to
  a steady state $\C = \lim_{t \to \infty}\C(t)$ such that
  \begin{equation}
    \C =  \A^2\T + \sigma^2 \P \ones\ones^\tr \P^\tr + \P \C \P^\tr.
  \label{eq:wait-recurrence}
  \end{equation}
  In particular, $\C$ is independent of the starting estimators $\bfy(0)$.
}

\copyableTheorem{BestResponseNonOptimality}{

  Let $G$ be a graph with $[n]$ vertices. Fix $\sigma$ and
  $\{\tau_i\}_{i \in [n]}$.

  Consider {\em best response} dynamics for $n$ agents on $G$ with
  $\sigma$ and $\{\tau_i\}_{i \in [n]}$. Let $\C^{br}$ denote the
  steady state the system converges to.

  Consider {\em fixed response} dynamics with some $\P$ and $\A$ for
  $n$ agents on $G$ with $\sigma$ and $\{\tau_i\}_{i \in [n]}$. Let
  $\C^{fr}$ denote the steady state the system converges to.

  Then there exists a choice of $n$, $G$, $\sigma$, $\{\tau_i\}$, $\A$
  and $\P$ such that $C_{ii}^{br} > C_{ii}^{fr}$ for all $i \in [n]$.
}

\copyableTheorem{NoAysmptoticLearningInFixedResponse}{ If
  $\sigma,\tau>0$, no fixed response dynamics can achieve socially
  asymptotic learning.  }

\copyableTheorem{PenultimatePerfect}{
   Penultimate prediction on the complete graph achieves perfect learning.
 }

\section{Background results}
\label{'section background results'}
\subsection{Time evolution of the covariance matrix}
We commence by proving a preliminary proposition on the relation
between the coefficients matrices $\P(t)$ and $\A(t)$, and the
covariance matrix $\C(t)$ in the best response and fixed response
models. This result does not depend on how $\P(t)$ and $\A(t)$ are
calculated, and therefore applies to both models.

First, let us calculate the covariance matrix directly. By the
definition of $\C(t)$ and by Eq.~\eqref{eq:y-def-vec}
we have that
\begin{align*}
  \C(t) = \Var{\Y(t)-\ones S(t)} = \Var{\A(t) \M(t) + \P(t) \Y(t-1) -\ones S(t)}.
\end{align*}
Since $S(t) = S(t-1)+X(t-1)$ then we can write
\begin{align*}
  \C(t) = \Var{\A(t) \bigl(\M(t)-\ones S(t)\bigr) + \P(t) \Y(t-1) - (\I-\A(t))\ones\bigl(S(t-1)+X(t-1)\bigr)}.
\end{align*}
Since the estimators $\{Y_i(t)\}$ are unbiased then $\P(t)\ones =
(\I-\A(t))\ones$; see the definitions of the models in
Section~\ref{'section-models'}, and in particular
Eq.~\eqref{eq:convex}. Hence
\begin{align*}
  \C(t) = \A(t) \T\, \A(t)^\tr  + \Var{\P(t)(\Y(t-1) - \ones S(t-1))} + \Var{\P(t)\ones X(t-1)},
\end{align*}
since $\Var{\M(t)-\ones S(t)} = \Var{\D(t)} = \T$. Finally, since
$\Var{\bfy(t-1)} = \C(t-1)$ we can write
\begin{align}
  \C(t) = \A(t)^2 \T + \P(t)\C(t-1)\P(t)^\tr +
  \sigma^2\P(t)\ones\ones^\tr\P(t)^\tr. \label{eq:recurrence}
\end{align}

\begin{proposition}
  Let $\Q(r,t) = \prod_{s=r+1}^t \P(s)$ and $\W(t) = \A(t)^2\T +
  \sigma^2 \P(t) \ones\ones^\tr \P(t)^\tr$ with $\W(0) = \C(0)$.  Then
  for all $t \geq 1$,
  \begin{align}
    \label{eq:c-closed-form}
    \C(t) = \sum_{r=0}^{t} \Q(r,t) \W(r) \Q(r,t)^\tr.
  \end{align}
  \label{prop:c-closed-form}
\end{proposition}
\begin{proof}
  First note that equation~\eqref{eq:recurrence} becomes simply $\C(t)
  = \W(t) + \P(t) \C(t-1) \P(t)^\tr$.  The base of the induction $t=1$
  is now immediate, since $\W(0)=\C(t-1)$. Now assume that it holds to
  time $t$. Then
  \begin{align*}
    \C(t+1)
    &= \W(t+1) + \P(t+1) \C(t) \P(t+1)^\tr
    \\
    &= \W(t+1) + \P(t+1) \left[\sum_{r=0}^{t} \Q(r,t) \W(r) \Q(r,t)^\tr \right] \P(t+1)^\tr
    \\
    &= \W(t+1) + \sum_{r=0}^{t} \Q(r,t+1) \W(r) \Q(r,t+1)^\tr
    \\
    &= \sum_{r=0}^{t+1} \Q(r,t) \W(r) \Q(r,t)^\tr.
  \end{align*}
\end{proof}

\subsection{Minimum variance unbiased linear estimator}
\label{section:min_var_est}
We show in this subsection how in general a minimum variance unbiased
linear estimator is calculated, given a collection of estimators with
a known covariance matrix.

Let $X$ be a random variable and let $(Z_1, \ldots, Z_n)$ be random
variables such that $\CondE{Z_i}{X}=X$ for all $i \in [n]$. Let
$C_{ij} = \Cov{Z_i-X,Z_j-X}$, with $\C$ being the matrix with entries
$C_{ij}$.

Let $M = \sum_ib_iZ_i$ be the minimum variance unbiased linear
estimator of $X$, i.e., let $(b_1,\ldots,b_n)$ minimize $\Var{M - X}$
under the constraint that $\CondE{M}{X} = X$, which is equivalent to
$\sum_ib_i=1$, since $\CondE{Z_i}{X}=X$ for all $i$.

Denote $\B = (b_1,\ldots,b_n)$.
\begin{proposition}
  \label{prop:minimum-var-est}
  \begin{align}
    \label{eq:minimum-var-est}
    \B = \frac{\ones^\tr \C^{-1}}{\ones^\tr \C^{-1} \ones}.
  \end{align}
\end{proposition}
\begin{proof}
  By definition
  \begin{align*}
    \Var{M-X} = \Var{\sum_ib_iZ_i-X} = \Cov{\sum_ib_iZ_i-X, \sum_jb_jZ_j-X}.
  \end{align*}
  Since $\sum_ib_i=1$ then we can write
  \begin{align*}
    \Var{M-X} = \Cov{\sum_ib_i(Z_i-X), \sum_jb_j(Z_j-X)},
  \end{align*}
  and then by the bilinearity of covariance we have that
  \begin{align*}
    \Var{M-X} = \sum_{ij}b_ib_j\Cov{Z_i-X,Z_j-X} = \B^\tr \C \B.
  \end{align*}
  Note that we again used here the fact that $\sum_ib_i=1$.

  To find $\B$ that minimizes this expression under the constraint
  that $\sum_ib_i=1$ we use Lagrange multipliers to minimize
  \begin{align*}
    \B^\tr \C \B + \lambda(\ones^\tr\B-1),
  \end{align*}
  which is a straightforward calculation yielding
  Eq.~\eqref{eq:minimum-var-est}.
\end{proof}
We assumed in this proof that $\C$ is an invertible matrix. When $\C$
is not invertible then it is easy to show that the same statement
holds, with $\C^{-1}$ being a pseudo-inverse of $\C$. While for
different such pseudo-inverses one gets different values of $\B$, the
variance of the different $M$'s is identical.

The following two corollaries follow directly from
Proposition~\ref{prop:minimum-var-est}.
\begin{corollary} \label{'corollary weight is inverse of variance'} If
  $\C$ is a diagonal matrix, then $b_i$, the weight given to each
  variable $Z_i$, in the minimum variance unbiased estimator is
  proportional to $\Var{Z_i - X}^{-1}$ and the variance of the minimal
  variance unbiased estimator is $1/(\sum_i \Var{Z_i - X}^{-1})$.
\end{corollary}

\begin{corollary} \label{'corollary combining two independent values'}
  If $$\C = \left( \begin{array}{cc} \sigma_1^2 & 0 \\ 0 &
      \sigma_2^2 \end{array} \right)$$ then the minimum variance
  unbiased estimator is
  \begin{align*}
   M = \frac{\sigma_2^2}{\sigma_1^2 +
    \sigma_2^2}Z_1 + \frac{\sigma_1^2}{\sigma_1^2 + \sigma_2^2}Z_2,
  \end{align*}
  with
  \begin{align*}
   \Var{M-X} = \frac{\sigma_1^2 \sigma_2^2}{\sigma_1^2 +
    \sigma_2^2}.
  \end{align*}
\end{corollary}

Note that in the best response model $Y_i(t)$ is the minimum variance
unbiased linear estimator of $S(t)$ given $M_i(t)$ and
$\{Y_j(t-1)\}$. Hence to calculate it is suffices to know the
covariances of these estimators.  It follows from the definitions that
\begin{align*}
  \Cov{Y_j(t-1)-S(t),Y_k(t-1)-S(t)} = \C_{jk}(t-1)+\sigma^2,
\end{align*}
\begin{align*}
  \Cov{Y_j(t-1)-S(t),M_i(t)-S(t)} = \tau_j^2+\sigma^2+\tau_i^2,
\end{align*}
and
\begin{align*}
  \Cov{M_i(t)-S(t),M_i(t)-S(t)} = \tau_i^2.
\end{align*}
Thus knowing $\C(t-1)$, $\sigma$ and $\{\tau_j\}$ is sufficient to
calculate the coefficients $A_i(t)$ and $\{P_{ij}(t)\}$ in the
best response model.

\subsection{Best response with a single agent}
\label{sec:kalman}
We provide the following proposition without proof. It is a
consequence of basic Kalman filter theory~\cite{kalman1960new}; it is
shown there that, for $n=1$, the MVULE of $S(t)$ given all the
measurements up to time $t$ is identical to the MVULE of $S(t)$ given
the new measurement at time $t$ and the previous estimator. Formally:
\begin{proposition}
  \label{prop:kalman-1}
  Best response achieves perfect learning when $n=1$.
\end{proposition}

\section{Best response dynamics}

\newcommand{\ba}{\boldsymbol{a}}
\renewcommand{\p}{{\bf p}}

Recall that in best response dynamics at time $t>0$ agent $i$ chooses
$A_i(t)$ and $\{P_{ij}(t)\}_j$ that minimize $C_{ii}(t)$, under the
constraints that $P_{ij}(t)=0$ if $j \not \in \partial i$, and
$A_i(t)+\sum_jP_{ij}(t)=1$. Thus $Y_i(t)$ is in fact the MVULE (see
definition in Section~\ref{subsec:best-response}) of $S(t)$, given
$M_i(t)$ and $\{Y_j(t-1)\}_{j \in \partial i}$.

Note that to calculate $A_i(t)$ and $\{P_{ij}(t)\}_j$ it is
necessary (and, in fact, sufficient, as we note in
Section~\ref{section:min_var_est}) to know $\C(t-1)$, $\sigma$
and $\{\tau_j\}$, and so this model indeed assumes that the agents know
the covariance matrix of their neighbors' estimators.

\subsection{Understanding best-response dynamics}
\label{section:understanding-best-response}

The condition that estimators are unbiased, or
$A_i(t)+\sum_jP_{ij}(t)=1$, means that given $\{P_{ij}(t)\}_j$ one can
calculate $A_i(t)$, or alternatively given $\P$ one can calculate
$\A$. Hence, fixing $\sigma$ and $\{\tau_j\}$, $\P(t)$ is a
deterministic function of $\C(t-1)$. Since by
Eq.~\eqref{eq:recurrence} $\C(t)$ is a function of $\A(t)$, $\P(t)$ and
$\C(t-1)$, then under best response dynamics, $\C(t)$ is in fact a
function of $\C(t-1)$.  We will denote this function by $F$, so that
$\C(t) = F(\C(t-1))$.  Our goal is to understand this map $F$, and in
particular to determine its limiting behavior.

We next analyze in more detail the best response calculation for agent
$i$.  This can conceptually be divided into two stages: calculating a
best estimator for $S(t)$ from $\Y(t-1)$, and then combining that with
$M_i(t)$ for a new estimator of $S(t)$.

Let the vector $\Y_{\partial i}(t-1) = \{Y_j(t-1)|j\in \partial i\}$
and let $\C_i(t-1) = C_{\partial i,\partial i}(t-1)$ be the covariance
matrix of the estimators of the neighbors of agent $i$.

Denote by $\q_{i}(t)$ the vector of coefficients for $\Y_{\partial
  i}(t-1)$ that make $Z_i = \q_i(t)^\tr \Y_{\partial i}(t-1)$ a
minimum variance unbiased linear estimator for $S(t)$; note that this
is also the estimator for $S(t-1)$. Then by
Proposition~\ref{prop:minimum-var-est} we have that
\begin{align*}
  \q_{i}(t) = \beta_i(t-1) \ones^\tr \C_i(t-1)^{-1},
\end{align*}
where $\beta_i(t-1) = 1/\ones^\tr \C_i(t-1)^{-1}\ones$.
It is easy to see that $\Var{Z_i-S(t-1)} = \beta_i(t-1)$ and thus
$\Var{Z_i-S(t)} = \beta_i(t-1) + \sigma^2$.

$M_i(t)$ is an independent estimator of $S(t)$ with variance
$\tau_i^2$. To combine it optimally with $Z_i$ we set
\begin{align}
  A_i(t) = \frac{\beta_i(t-1) +
    \sigma^2}{\tau_i^2+\beta_i(t-1)+\sigma^2} \geq
  \frac{\sigma^2}{\tau_i^2+\sigma^2},
  \label{eq:alpha}
\end{align}
by Corollary~\ref{'corollary combining two independent values'}.  The
optimal weight vector $\bp_i(t)$ for agent $i$ (i.e., $\{P_{ij}\}_{j
  \in \partial i}$) is therefore $\bp_i(t) = (1-A_i(t))\q_i(t)$.

\subsection{Complete graph case}

When $G$ is the complete graph, the agents best-respond similarly,
since they all observe the same set of estimators from the previous
iteration. We now have $\C_i(t-1) = \C(t-1)$, $\q_i(t) = \q(t)$, and
$\beta_i(t-1) = \beta(t-1)$, for all $i$.  For the moment, we will
suppress the $t$.  Letting $\ba$ be the vector with coefficients $A_i$,
we then have $\P = (\ones-\ba)\q^\tr = \beta(\ones-\ba)\ones^\tr
\C^{-1}$.  Using this form for $\P$, we can now see that $\P\C\P^\tr =
\beta(\ones-\ba)(\ones-\ba)^\tr$.  Putting this all together, and
adding back the $t$, we have by Eq.~\eqref{eq:recurrence} that
\begin{align}
  \C(t) = \A(t)^2\,\T + (\beta(t-1)+\sigma^2)(\ones-\ba(t))(\ones-\ba(t))^\tr.
  \label{eq:clique-c}
\end{align}

Since by equation~\eqref{eq:alpha}, $A_i(t)$ depends only on
$\beta(t-1)$, $\tau_i$, and $\sigma$, we see that $\C(t) = F(\C(t-1))$
depends on $\C(t-1)$ only through $\beta(t-1) = 1/\ones^\tr
\C(t-1)^{-1}\ones$.  Hence we can write $\C(t) = \C(\beta(t-1))$.  We
now see that we can completely describe the state of the system by a
single parameter $\beta$, and our map $F$ reduces to the map
$f:\beta\mapsto1/\ones^\tr \C(\beta)^{-1}\ones$.  We wish to analyze
this function $f$ as a single-parameter discrete dynamical system.


To simplify our formula for $f$, we will make use of a matrix identity
attributed to Woodbury and others (cf.~\cite{higham2002accuracy}):
\begin{theorem}[Sherman-Morrison-Woodbury formula]
  For any $U \in \R^{n\times k}$, $V \in \R^{k\times n}$, and
  nonsingular $X \in \R^{n\times n}$, $Y \in \R^{k\times k}$ such that
  $Y^{-1}+VX^{-1}U$ is nonsingular,
  \begin{equation}
    \label{eq:woodbury}
    \left(X+UYV \right)^{-1} = X^{-1} - X^{-1}U \left(Y^{-1}+VX^{-1}U \right)^{-1} VX^{-1}.
  \end{equation}
\end{theorem}

Using the formula in Eq.~\eqref{eq:woodbury}, we can expresses $f$ in terms
of $\beta$:
\begin{lemma}
  Let $y = \sum_i \tau_i^2$ and $z = (\sum_i \tau_i^2)(\sum_i
  \tau_i^{-2})$.  Then $f(\beta)$ has the following form:
  \begin{align}
    f(\beta) &=
    \frac
    {y \left(\beta +\sigma ^2\right) \left(y+\beta +\sigma ^2\right)}
    {y \left(y-(n-2) n \left(\beta +\sigma ^2\right)\right)+z
      \left(\beta +\sigma ^2\right) \left(y+\beta +\sigma ^2\right)}
  \end{align}
  \label{lem:clique-f}
\end{lemma}
\begin{proof}
  Let $x=(\beta+\sigma^2)$ for brevity.  We will compute $\ones^\tr
  \C(\beta)^{-1}\ones$ by applying the matrix identity
  \eqref{eq:woodbury} to equation~\eqref{eq:clique-c}, with $k=1$, $X
  = \A^2\,\T$, $Y = I_{1\times 1}$, and
  $U^\tr=V=\sqrt{x}(\ones-\ba)$.  This gives us:
  \begin{align}
    \ones^\tr C(\beta)^{-1} \ones
    &= \ones^\tr \left(\A^{-2}\,\T^{-1} - x \frac{\A^{-2}\,\T^{-1}(\ones-\ba)(\ones-\ba)^\tr \A^{-2}\,\T^{-1}}
      {1 + x(\ones-\ba)^\tr \A^{-2}\,\T^{-1} (\ones-\ba)}\right) \ones
    \nonumber
    \\
    &= \ones^\tr \A^{-2}\,\T^{-1} \ones - x \frac{\left(\ones^\tr \A^{-2}\,\T^{-1}(\ones-\ba)\right)^2}
    {1 + x(\ones-\ba)^\tr \A^{-2}\,\T^{-1} (\ones-\ba)} \ones.
    \label{eq:1cinv1}
  \end{align}
  We have the identities
  \begin{align}
    & \ones^\tr \A^{-2}\,\T^{-1} (\ones-\ba) \;=\; nx + y x^2,
    \nonumber
    \\
    & \ones^\tr \A^{-2}\,\T^{-1} \ones \;=\; y x^2 + 2 n x + z/y,
    \nonumber
    \\
    & (\ones-\ba)^\tr \A^{-2}\,\T^{-1} (\ones -\ba) \;=\; y x^2,
  \end{align}
  from the expression~\eqref{eq:alpha} for $\ba$.  These identities
  allow us to simplify equation~\eqref{eq:1cinv1}:
  \begin{align*}
    \ones^\tr C(\beta)^{-1} \ones &= y x^2 + 2 n x +
    \frac z y - \frac {(n + y x)^2} { x^{-1} +
      y }.
  \end{align*}
  Finally, setting $f(\beta) = 1/\ones^\tr C(\beta)^{-1} \ones$ and
  simplifying gives us the result.
\end{proof}

We are now ready to prove the main theorem of this section.
\thmBestResponseSteadyState
\begin{proof}
  We will make use of the Banach fixed point
  theorem~\cite{banach1922operations} which states that if there
  exists some $k<1$ such that $|f'(\beta)|<k$ for all $\beta$, then
  there is a unique fixed point $\beta^*$ of $f$, and iterates of $f$
  satisfy $|f^t(\beta)-\beta^*| < \frac{k^t}{1-k}|\beta-f(\beta)|$ for
  all starting points $\beta$.  Thus, given this theorem, we need only
  show $|f'(\beta)|<k$ for some $k < 1$.

  First note that the $n=1$ case reduces to a Kalman filter, which we
  review briefly in Section~\ref{sec:kalman}.

  We will find it convenient to think of horizontal shift $g(x) =
  f(x-\sigma^2)$, where $x=\beta+\sigma^2$, but allow any $x>0$.  That
  is, $g$ can be thought of as taking the variance $x$ of the estimate
  of the process \emph{this} round using only estimates from last
  round.  We first compute $g$ and its first two derivatives; letting
  $D = y (y-(n-2) n x)+x z (x+y)$, we have
  \begin{align}
    g(x) & = x y (x+y) / D \label{eq:gx}\\
    g'(x) & = y^2 (y-(n-2) x) (n x+y) / D^2 \label{eq:dgx}\\
    g''(x) & = 2 y^2 \left((n-2) n x^3 z+y^3 \left((n-1)^2-z\right)-3 x^2 y z-3 x y^2 z\right) / D^3 \label{eq:ddgx}
  \end{align}

  Before proving bounds on $f'$, we prove a few useful observations:
  \begin{enumerate}
  \item[A.] $z\geq n^2$.  This follows from the Cauchy-Schwarz inequality,
    since $\sum_i \tau_i\tau_i^{-1} = n$.
  \item[B.] $D>0$.  Expanding $D$, we have three strictly positive
    terms $2 n x y + y^2 + x^2 z$ plus $x y z - n^2 x y$, which is
    non-negative by observation A.
  \item[C.] $g'(x) < 0 \iff y < (n-2) x$.  Since $D>0$ by observation B, this follows from~\eqref{eq:dgx}.
  \end{enumerate}

  \textbf{Upper bound:} We will show that there exists some $k_U$ such
  that $f'(\beta)<k_U<1$ for all $\beta>0$ by showing $g'(x)<k_U$ for
  all $x>\sigma^2$.  First, note that $g'(0)=1$. Next, by observation
  C and some simple algebra, one can show that $g''(x)<0$ as long as
  $g'(x)>0$; that is, $g'$ is strictly decreasing while positive,
  until $x = (n-2)/y$.  Hence, letting $k_U := f'(0)$, we have $k_U =
  g'(\sigma^2) < g'(0) = 1$.  Finally, we have $f'(\beta) < f'(0) =
  k_U$ for all $\beta>0$, so $k_U$ is our upper bound.

  \textbf{Lower bound:} To show the lower bound, we minimize $g'$ with
  respect to $x$ as well as all the parameters.  Let us start with
  $z$.  By observation C, we know that the minimum value of $g'(x)$
  must occur when $x > (n-2)/y$, the region where $g'(x)<0$ (note that
  if $n=2$ we get a trivial lower bound of 0, so henceforth we will
  assume $n>2$).  In this region, it is clear from observations A and
  B that the minimum of $g'$ with respect to $z$ occurs when $z=n^2$.
  Substituting and simplifying, we have
  \begin{equation*}
    g'(x) \geq \frac{y^2 (y-(n-2) x)}{(n x+y)^3} := h(n,x,y).
  \end{equation*}
  We next minimize over $x$: solving $\frac{d}{dx}h(n,x,y)=0$ for $x$
  yields $x = y \frac{2n-1}{n(n-2)}$, and one can check that indeed
  $\frac{d^2}{dx^2}h(n,x,y)>0$ for this $x$.  Substituting again, we
  are now left with
  \begin{equation*}
    g'(x) \geq -\frac{(n-2)^3}{27 n (n-1)^2},
  \end{equation*}
  which is now only a function of $n$.  One can now easily see that
  $g' > -\frac{1}{27}$ for all $n,x,y,z$.

  We have therefore shown that $|f'(\beta)| < k := \max(f'(0), 1/27) <
  1$ for all $\beta$ and all parameter values.
\end{proof}

As a concluding comment we analyze the steady-state $\beta^*$.  From
the form of $f$, one can show that $\beta^*$ satisfies:
\begin{align}
  0 &= y \sigma ^2 \left(y+\sigma ^2\right)
  -\sigma ^2 \left(y (-(n-2) n+z-2)+z \sigma ^2\right) \beta
   \nonumber\\&\quad
  -\left((n-1)^2 y-z \left(y+2 \sigma ^2\right)\right) \beta^2
  -z\beta^3
  \label{eq:steady-state-cubic}
\end{align}

As a corollary of Theorem~\ref{thmBestResponseSteadyState}, this cubic
polynomial has a unique positive root.
\label{'section best-response'}
\section{Fixed response dynamics}
\label{'section fixed-response'}
Recall that in fixed response dynamics each agent $i$, at each round
$t$, has access to its neighbors' estimators from the previous round,
$\{Y_j(t-1)| j \in \partial i\}$, as well as its current measurement
$M_i(t)$. The new estimates are $\bfy(t) = \A \cdot \M(t)+\P \cdot
\bfy(t-1)$, i.e., fixed convex linear combinations of these values.

We first show the system converges to a steady state: as $t$ tends to
infinity the covariance matrix $\C(t) = \Var{\Y(t)-\ones S(t)}$ tends
to some matrix $\C$.

Note that as in Section~\ref{'section best-response'} above, a result
of the convexity condition is that the choice of $\P$ uniquely
determines $\A$.

\subsection{Convergence of fixed response dynamics}
To prove our theorem we shall need the following lemma, as well as an
easy corollary of Proposition~\ref{prop:c-closed-form}.

\begin{lemma}
  Let $P_i \in \R^{n\times n}_+$ satisfy
  $\|P_i\|_\infty \leq \gamma$ for all $1\leq i\leq t$.  Then for any
  $Q \in \R^{n\times n}_+$
  \begin{align*}
    \left\|\left(\prod_{i=1}^t P_i\right) Q \left(\prod_{i=t}^1 P_i^\tr\right)\right\|_\infty \leq n\gamma^{2t}\|Q\|_\infty
  \end{align*}
  \label{lem:matrix-prod}
\end{lemma}
\begin{proof}
  This follows from two facts about the infinity norm.  First,
  $\|\cdot\|_\infty$ is submultiplicative, meaning for all $A,B \in
  \R^{n\times n}$, $\|AB\|_\infty \leq \|A\|_\infty \|B\|_\infty$.
  Second, for all $A$ we have $\|A\|_\infty \leq n\|A^\tr\|_\infty$.
  These two combined give us
  \begin{align*}
     \left\|\prod_{i=1}^t P_i\right\|_\infty \leq \gamma^t \quad \text{and} \quad
     \left\|\prod_{i=t}^i P_i^\tr\right\|_\infty \leq n\left\|\prod_{i=1}^t P_i\right\|_\infty \leq n\gamma^t,
  \end{align*}
  and the result then follows from submultiplicativity.
\end{proof}

\begin{corollary}
  \begin{align}
    \C(t) = \P^t \C(0) {\P^\tr}^t + \sum_{r=0}^{t-1}\P^r (\A^2 +
    \sigma^2P \ones\ones^\tr \P^\tr) {\P^\tr}^r.
  \end{align}
\end{corollary}

The main theorem of this subsection is the following.
\thmFixedResponseConverges
\begin{proof}

  Let $\gamma = 1 - \max_i A_i < 1$.  Then since the entries of $\P$
  are non-negative, the absolute row sums of $\P$ are less than
  $\gamma$, so we have $\|\P\|_\infty \leq \gamma$.  Letting $\mZ =
  \A^2 + \sigma^2 \P \ones\ones^\tr \P^\tr$, we have by
  Lemma~\ref{lem:matrix-prod} that
  \begin{align*}
    \|\C(t+1) - \C(t)\|_\infty
    &= \left\|\sum_{r=0}^{t}\P^r X {\P^\tr}^r + \P^{t+1} \C(0) {\P^\tr}^{t+1}
                - \sum_{r=0}^{t-1}\P^r \mZ {\P^\tr}^r - \P^t \C(0) {\P^\tr}^t \right\|_\infty
    \\
    &= \left\|\P^t (\mZ + \P \C(0) \P^\tr - \C(0)) {\P^t}^\tr \right\|_\infty
    \\
    & \leq n\gamma^{2t} \|\mZ + \P \C(0) \P^\tr - \C(0)\|_\infty.
  \end{align*}
  Thus since $\gamma<1$, we have $\lim_{t\to\infty} \|\C(t+1) -
  \C(t)\|_\infty = 0$.  Moreover, for all $t$ we have
  \begin{align*}
    \|\C(t)\|_\infty
    &= \left\|\sum_{r=0}^{t-1}\P^r \mZ {\P^\tr}^r + \P^{t} \C(0) {\P^\tr}^{t}\right\|_\infty
    \\
    &\leq \left(\|\mZ\|_\infty + \|\C(0)\|_\infty\right)\sum_{r=0}^{t}n\gamma^{2r}
    \\
    &\leq n\frac{\|\mZ\|_\infty + \|\C(0)\|_\infty}{1-\gamma^2} < \infty,
  \end{align*}
  so clearly $\lim_{t \to \infty}\|\C(t)\|_\infty < \infty$.  Thus,
  $\lim_{t \to \infty}\C(t)$ exists, and \eqref{eq:wait-recurrence}
  follows from the recurrence in equation~\eqref{eq:recurrence}.

  To see that the choice of $\C(0)$ is immaterial, consider the
  alternate sequence $\tilde \C(t)$ resulting from another choice
  $\tilde \C(0)$ (but the same $\P$).  By definition,
  \begin{align*}
    \tilde \C(t) = \sum_{r=0}^{t-1}\P^r \mZ {\P^\tr}^r + \P^t \tilde \C(0) {\P^\tr}^t.
  \end{align*}
  By Lemma~\ref{lem:matrix-prod} we have that
  \begin{align*}
    \|\C(t) - \tilde \C(t)\|_\infty = \|\P^t (\C(0)-\tilde \C(0)) {\P^\tr}^t\|_\infty
    \leq n\gamma^{2t}\|\C(0)-\tilde \C(0)\|_\infty,
  \end{align*}
  so we clearly have $\lim_{t\to\infty}\|\C(t) - \tilde \C(t)\|_\infty =
  0$.  Thus, $\lim_{t\to\infty} \tilde \C(t) = \lim_{t\to\infty} \C(t) =
  \C$.
\end{proof}

\subsection{Non-optimality of best response steady-state}

By Theorem~\ref{thmBestResponseSteadyState} best response dynamics
converges to a unique steady state.  The next result shows that
although, in best responding, agents minimize the variance of their
estimators, in some cases they can converge to a steady state with
lower variances by an appropriate choice of a fixed response. I.e., by
cooperating the agents can achieve better results than by greedily
choosing the short-term minimum.

We consider the case of the complete graph over $n$ players where the
agents measurement errors $\tau_i$, are the same and equal $\tau = 1$,
and also the standard deviation of the state's random walk $\sigma = 1$.

Before proving the main theorem of this subsection we establish the
following lemma.
\begin{lemma}
  \label{'lemma fixed response variance'}
  Let $\C_{(n)}$ be the steady state of {\em fixed response} dynamics on the
  complete graph with $n$ agents, $\tau_i=\sigma=1$, $A_i=\alpha$ for
  all $i$ and $P_{ij} = (1-\alpha)/n$ for all $i$ and $j$. Then
  \begin{align}
    \label{'equation calc variance from fixed alpha'}
    C_{ii}(t) = \alpha^2 + \frac{(1-\alpha)^2(1 + \alpha^2/n)}{(2-\alpha)\alpha}.
  \end{align}
\end{lemma}
\begin{proof}
  Let $\beta^{\alpha}_n = 1/\ones^\tr\C_{(n)}^{-1}\ones$, and let $Z_n(t) =
  \frac{1}{n}\sum_i Y_i(t)$, so that
  \begin{align*}
    \beta_n^{\alpha} = \lim_{t \to \infty} \Var{Z_n(t) - S(t)}.
  \end{align*}
  By the symmetry of the model we have that $Y_i(t+1) = (1 -
  \alpha)Z_n(t) + \alpha
  M_i(t+1)$, and so
  \begin{align*}
    Z_n(t+1) = (1 - \alpha)Z_n(t) + \alpha\frac{1}{n}\sum_i M_i(t+1).
  \end{align*}
  Therefore
  \begin{align*}
    \beta_n^{\alpha} = \Var{Z_n(t+1) - S(t+1)} =
    \frac{\alpha^2}{n}\tau^2 + (1 - \alpha)^2 (\beta^{\alpha}_n +
    \sigma^2),
  \end{align*}
  which, since $\tau=\sigma=1$, implies by simple manipulation that
  \begin{align*}
    \beta_n^{\alpha} = \frac{\alpha^2/n+(1-\alpha^2)}{1-(1-\alpha^2)}.
  \end{align*}
  Finally, because $Y_i(t) = \alpha M_i(t) + (1 - \alpha)Z_n(t-1)$ then
  \begin{align*}
    C_{ii}(t) &= \lim_{t \to \infty}\Var{Y_i(t) - S(t)} = \alpha^2 + (1
    - \alpha)^2 (\beta_n^{\alpha} + 1)\\
    &= \alpha^2 + \frac{(1-\alpha)^2(1 + \alpha^2/n)}{(2-\alpha)\alpha}.
  \end{align*}
\end{proof}

\thmBestResponseNonOptimality
\begin{proof}
  Let $n=2$, let $G$ be the complete graph on two vertices, let
  $\sigma=1$ and let $\tau_i=1$ for all $i \in [n]$. Let $A_i=\alpha$
  for all $i$ and $P_{ij} = (1-\alpha)/n$ for all $i$ and $j$. Then by
  Lemma~\ref{'lemma fixed response variance'} we have that
  \begin{align*}
    C_{ii}^{fr} &= \alpha^2 + \frac{(1-\alpha)^2(1 +
      \alpha^2/2)}{(2-\alpha)\alpha}.
  \end{align*}
  for all $i$.

  By Eq.~\eqref{eq:steady-state-cubic} we have that
  \begin{align*}
    C_{ii}^{br} = 2-\sqrt{2} \approx 0.58578,
  \end{align*}
  for all $i$.

  It is easy to verify that for $\alpha=0.60352$, for example (which
  is in fact the minimum), it holds that $C_{ii}^{fr} \approx
  0.58472$ and so $C_{ii}^{br} > C_{ii}^{fr}$.
\end{proof}

We note that the choice $n=2$ was made to make the proof above simple,
rather than being a pathological example; we now show that the same
holds for $n$ large enough.
\begin{lemma}
  \label{'lemma inf best response'}
  Let $\C^{br,(n)}$ be the steady state of
  {\em best response} dynamics on the complete graph with $n$ agents
  and $\tau_i=\sigma=1$. Then
  \begin{align*}
    \lim_{n \to \infty}C^{br,n}_{ii} = \frac{\beta +1}{\beta +2},
    \quad\quad\mathrm{  where }\quad\quad
    \beta = {\tfrac 2 3} \sqrt{7} \cos\left({\tfrac 1 3} \tan^{-1}(3
      \sqrt{3})\right)-{\tfrac 4 3}.
  \end{align*}
\end{lemma}
We omit the proof and mention that it follows directly by substitution
into, and solution of, the cubic polynomial of
Eq.~\eqref{eq:steady-state-cubic}, and Eq.~\eqref{eq:alpha}. We
likewise omit the proof of the following lemma, which is an immediate
corollary of Lemma~\ref{'lemma fixed response variance'} above.
\begin{lemma} \label{lemma best response suboptimal for large n}
  Let $\C^{fr,n,\alpha}$ be the steady state of {\em fixed response}
  dynamics on the complete graph with $n$ agents, $\tau_i=\sigma=1$,
  $A_i=\alpha$ for all $i$, and $P_{ij} = (1-\alpha)/n$ for all $i$
  and $j$. Then
  \begin{align*}
    \lim_{n \to \infty}C^{fr,n,\alpha}_{ii} =
    \alpha^2+\frac{1}{\alpha(2-\alpha)}-1.
  \end{align*}
\end{lemma}
Hence setting $\alpha_{\infty} = 0.59075$ (again the minimum) we get
\begin{align*}
  \lim_{n \to \infty}C^{fr,n,\alpha_\infty}_{ii} \approx 0.55017,
\end{align*}
and using Lemma~\ref{'lemma inf best response'} it is easy to numerically
verify that
\begin{align*}
  \lim_{n \to \infty}C^{br,n}_{ii}  \approx 0.55496.
\end{align*}
Thus we have shown that for large enough $n$ it again holds that
$C_{ii}^{br} > C_{ii}^{fr}$ on the complete graph, for the appropriate
choice of parameters. We conjecture that this is in fact achievable
for all $n>1$, with the correct choice of parameters. In the case of
$n = 1$ best response is optimal by Proposition~\ref{prop:kalman-1}.

\subsection{Socially asymptotic learning}

Recall that in the complete graph setting with fixed $\sigma$ and
$\tau$ and $\tau_i < \tau$ for all $i \in [n]$ we say that a dynamics
is socially asymptotically learning if the variance of each agent's
estimator approaches $\frac{\sigma^2 \tau_i^2}{\sigma^2 + \tau_i^2}$
as the number of agents increases.

Surprisingly, no fixed response dynamics can achieve socially
asymptotic learning unless either $\sigma = 0$ or $\tau = 0$.  This,
of course, includes the steady state of the best response dynamics.
In the case that $\sigma = 0$ the value of $S(t)$ is constant over
time, and we are in the DeGroot model which is known to converge.  In
the case that $\tau = 0$ each agent receives the exact value of $S(t)$
in each round and can simply set $Y_i(t) = M_i(t)$ to asymptotically
learn.

\thmNoAysmptoticLearningInFixedResponse

\begin{proof}
  For the sake of contradiction, first assume that there exists some
  fixed response scheme that permits socially asymptotic learning so
  that in the limit $C_{ii} = \frac{\sigma^2 \tau_i^2}{\sigma^2 +
    \tau_i^2}$ for all $i \in [n]$.  We analyze
  Equation~\ref{eq:recurrence} which states that: $\C(t) = \A(t)^2 \T
  + \sigma^2\P(t)\ones\ones^\tr\P(t)^\tr + \P(t)\C(t-1)\P(t)^\tr $.

  It is easy to see that each of the three terms on the right hand
  side of Equation~\ref{eq:recurrence} is a positive semidefinite
  matrix and thus will have non-negative diagonal entries.  By our
  assumption that we have a socially asymptotic learner, we know that
  the sum of the $i, i$ entries in the three matrices on the right
  hand side of Equation~\ref{eq:recurrence} equals $\frac{\sigma^2
    \tau_i^2}{\sigma^2 + \tau_i^2}$ in the limit.  Now,
  $(\P(t)\ones\ones^\tr\P(t)^\tr)_{i, i} = \sum_{j, k}p_{ij} p_{ik} =
  \sum_{j}p_{ij}\sum_{k} p_{ik} = (1 - \alpha_i)^2$.  By fixing
  $\epsilon_i$ so that $\alpha_i = \frac{\sigma^2 + \epsilon_i}{\sigma^2
    + \tau_i^2}$ we see that in the limit
  \begin{align}
     \frac{\sigma^2 \tau_i^2}{\sigma^2 + \tau_i^2} = C_{ii} & \leq (\A(t)^2 \T)_{i, i} + (\sigma^2\P(t)\ones\ones^\tr\P(t)^\tr)_{i, i} + (\P(t)\C(t-1)\P(t)^\tr )_{i, i} \nonumber \\
      & = \alpha_i^2 \tau_i^2 + (1 - \alpha_i)^2\sigma^2 + (\P(t)\C(t-1)\P(t)^\tr )_{i, i} \nonumber \\
     & = \frac{\sigma^2
    \tau_i^2}{\sigma^2 + \tau_i^2} +  \frac{\epsilon_i^2}{\sigma^2 + \tau_i^2} + (\P(t)\C(t-1)\P(t)^\tr )_{i, i} \end{align} and we have that in the limit for all $i \in [n]$: $\epsilon_i$ goes to 0, $(\P(t)\C(t-1)\P(t)^\tr )_{i, i}$ goes to 0, and $\alpha_i$ goes to $\frac{\sigma^2 }{\sigma^2 + \tau_i^2}$.

    We will now show that it cannot be that both $\alpha_i$ goes to $\frac{\sigma^2 }{\sigma^2 + \tau_i^2}$ and that $(\P(t)\C(t-1)\P(t)^\tr )_{i, i}$ goes to 0.  Another application of Equation~\ref{eq:recurrence} yields that \begin{align*} (\P(t)\C(t-1)\P(t)^\tr )_{i, i}  =\;& (\P(t)\A(t-1)^2 \P(t)^\tr)_{i, i} + \sigma^2(\P(t)\P(t-1)\ones\ones^\tr\P(t-1)^\tr\P(t)^\tr)_{i, i} \\&\quad + (\P(t)\P(t-1)\C(t-2)\P(t-1)^\tr\P(t)^\tr)_{i,i}\end{align*} and again all the matrices on the right hand side are positive semi-definite and thus have non-negative diagonals.  Thus:

  \begin{align*}
    (\P(t)\C(t-1)\P(t)^\tr )_{i, i}  & \geq \sigma^2(\P(t)\P(t-1)\ones\ones^\tr\P(t-1)^\tr\P(t)^\tr)_{i,i} \nonumber \\
    & = \sigma^2 \sum_{j, k, l, m} p_{ij}p_{jk}p_{il}p_{lm} = \sigma^2
    \left( \sum_{j} p_{ij} \sum_{k} p_{jk} \right)^2 = \sigma^2 \left(
      \sum_{j} p_{ij} (1 - \alpha_j) \right)^2 \nonumber \\ &\geq
    \sigma^2(1 - \alpha_i)^2 \min_j (1 - \alpha_j)^2 \geq
    \frac{\sigma^2\tau^8}{(\sigma^2 + \tau^2)^4} > 0,
  \end{align*}
   and this is a contradiction because we saw that $(\P(t)\C(t-1)\P(t)^\tr )_{i, i}$ limited to 0.  The penultimate inequality is because $\alpha_i$ limits to  $\frac{\sigma^2 }{\sigma^2 + \tau_i^2}$ and because $\tau_i \leq \tau$ for all $i$.
\end{proof}

\section{Penultimate prediction dynamics}
\label{'section memory 1'}

In this section we consider players that can remember one value from
the previous round.  We show that, in the case of the complete graph,
this allows the agents to learn substantially more efficiently than in
the previous models. Also on the complete graph we show that this
model features perfect learning.  In general, it is not clear that the
optimal strategy involving one remembered value must necessarily have
this property.

We call this model the penultimate prediction model because an agent
using it is effectively trying to re-estimate the value of the
underlying state in the previous round, disregarding its own new
measurement from the current round.  This may help discount the older
information that contributed to the prediction of each neighbor in the
previous round.

As defined in Section~\ref{'section-models'}, in this model each agent
$i$ does the following at time $t$: (a) agent $i$ first picks
$A_i$ and $\{P_{ij}\}_j$ that minimize $\Var{R_i(t)-S(t-1)}$
where $\br(t) = \A \cdot \br(t-1)+\P \cdot \bfy(t-1)$, and (b) agent $i$
then chooses $k_i(t)$ that minimizes $\Var{Y_i(t)-S(t)}$ where
$Y_i(t) = k_i(t) M_i(t) + (1 - k_i(t)) R_i(t)$.

 \subsection{Complete graph case}

 We show that the penultimate prediction model achieves perfect
 learning on the complete graph.  This means that agent $i$ learns
 $S(t)$ as if she had access to every agent's measurements from all the previous rounds, rather than just her neighbors' estimators from the last round.

\thmPenultimatePerfect

\begin{proof}
  Let
  \begin{align*}
   \tau_{*} = \left(\sum_{i \in
      [n]} \tau_i^{-2} \right)^{-\frac{1}{2}},
  \end{align*}
  let $q_i= \tau^2_{*}/\tau_i^2$ and let $\bar{M}(t) =
  \sum_iq_iM_i(t)$ be the average of the measurements of the agents at
  time $t$, weighted by the inverse of their variance. Then $\bar M(t)$
  is the MVULE of $S(t)$ given $\M(t)$ (see Corollary~\ref{'corollary
    weight is inverse of variance'}).

  Let $E(t)$ be the MVULE of $S(t)$, given {\em all that is known up
    to time $t$}: $\{\M(s)|s \leq t\}$, together with $\bfy(0)$.  Let
  $V(t) = \Var{E(t) - S(t)}$.

  Basic Kalman filter theory (see, e.g. \cite{kalman1960new}) shows
  that $E(t)$ can be written as
  \begin{align}
    \label{eq:E-kalman}
    E(t) = (1 - K(t))E(t-1) + K(t)\bar{M}(t),
  \end{align}
  with
  \begin{align*}
    V(t+1) = (V(t) + \sigma^2)(1 - K(t))
  \end{align*}
  and $K(t) = \frac{V(t) + \sigma^2}{V(t) + \sigma^2 + \tau_{*}^2}$.
  Note that $V(t)$ is deterministic.


  We now prove by induction that $R_i(t) = E(t-1)$.  The base case of
  $t = 1$ follows from definitions.

  By our inductive hypothesis at step $t$ we have that $R_i(t) =
  E(t-1)$ and $\Var{R_i(t) - S(t-1)} = V(t-1)$. Hence $R_i(t)$ is
  identical for all agents and we can write $R(t) = R_i(t) = E(t-1)$.
  Because $R(t) - S(t-1)$ and $S(t) - S(t-1)$ are independent  we have
  that
  \begin{align*}
    \Var{R(t)-S(t)} = \Var{R(t)-S(t-1)} + \sigma^2 = V(t) +
    \sigma^2.
  \end{align*}
  Since $R(t) - S(t)$ and $M_i(t) - S(t)$ are independent,
  $\Var{M_i(t) - S(t)} = \tau_i^2$, and since $Y_i(t)$ is the MVULE of $S(t)$ given $R(t)$ and $M_i(t)$, then by
  Corollary~\ref{'corollary combining two independent values'},
  $Y_i(t)$ will satisfy
  \begin{align}
    \label{eq:y-r}
    Y_i(t) = k_i(t) R(t) + (1-k_i(t)) M_i(t),
  \end{align}
  where $k_i(t) = \frac{\tau_i^2}{V(t-1) + \sigma^2 + \tau_i^2}$ is also
  deterministic. Thus $Y_i(t)$ is a deterministic function of $R(t)$
  and $M_i(t)$, and more importantly $M_i(t)$ is a deterministic
  linear combination of $Y_i(t)$ and $R(t)$. Since $R(t+1)$ is the
  MVULE of $S(t)$ given $R(t)$ and $\bfy(t)$, its variance
  is bounded from above by MVULE of $S(t)$ given $R(t)$ and
  $\M(t)$. But by Eq.~\ref{eq:E-kalman} we have that the optimum is
  achieved by $E(t)$, and so $R(t+1) = E(t)$.

  We have therefore established that $Y_i(t)$ is the MVULE of $S(t)$
  given $E(t-1)$ and $M_i(t)$, where $E(t-1)$ is the MVULE of $S(t-1)$
  given {\em all the measurements} up to time $t-1$. To complete the
  proof we note that by the Markov property of $S(t)$ this means that
  $Y_i(t)$ is the MVULE of $S(t)$ given all the measurements up to
  time $t$, together with $M_i(t)$.

\end{proof}

 \begin{corollary} \label{'corollary penultimate is asymptotic'}

   In the complete graph case, for any fixed $\sigma$ and $\tau$ where
   $\tau_i \leq \tau$ for all $i \in [n]$, the penultimate prediction
   heuristic is a socially asymptotic learner.
 \end{corollary}
 \begin{proof}
   By Theorem~\ref{thmPenultimatePerfect} we need only show that the
   optimal learner is socially asymptotic.  Fix some agent $i$.  At
   round $t$ if this agent is given $\M(t-1)$, then, by
   Corollary~\ref{'corollary weight is inverse of variance'}, he can
   predict $S(t-1)$ with $R_i(t)$ such that $\Var{R_i(t) - S(t-1)} =
   \tau_{*}^2 = \left(\sum_{i \in [n]} \tau_i^{-2} \right)^{-1}$.
   However, note that $\tau_*^2 \leq \tau^2/n$.  Agent $i$ can then compute
   $Y_i(t) = \frac{\sigma^2}{\sigma^2 + \tau_i^2}M_i(t) +
   \frac{\tau_i^2}{\sigma^2 + \tau_i^2}R_i(t)$ and it is easy to see
   that $\Var{Y_i(t) - S(t)} \leq \frac{\sigma^2 \tau_i^2}{\sigma^2 +
     \tau_i^2} + \frac{\tau_i^6}{n(\sigma^2 + \tau_i^2)^2}$, which
   approaches $\frac{\sigma^2 \tau_i^2}{\sigma^2 + \tau_i^2}$ as $n$
   grows.
 \end{proof}

\section{Conclusion}
This work can be seen as a study of natural extensions of the DeGroot
model to the setting where the value to be learned changes over time.
The most direct extension is the {\em fixed response} model.  Here we
show that while the estimate will keep moving with the true values,
its variance will converge to a fixed value.  However, in contrast to
the DeGroot model, the agents are continually receiving new
independent signals, and so have a reference point from which to
evaluate the validity of their neighbors' signals. This leads us to
propose the {\em best response} model.  We show that in the case of
the complete graph, best response dynamics will always converge to a
particular fixed response that is (myopically) optimal.  However, we
also show that it is not necessarily Pareto optimal amongst all fixed
responses.  Finally, we show that a simple strengthening of the model
to allow agents to remember one value can, in certain cases, lead to
much improved performance.  This can be seen not only as a critique of
fixed response dynamics as being too weak to capture natural dynamics,
but also as an interesting model to be studied more in its own right.

\bibliographystyle{abbrv}
\bibliography{bib}

\end{document}